\documentclass[a4paper]{article}
\usepackage[margin=3cm]{geometry}
\usepackage{graphicx}
\usepackage{epstopdf}

\usepackage{url}
\usepackage{color}

\usepackage[affil-it]{authblk}
\usepackage{amsthm}
\usepackage{amsmath}
\usepackage{amssymb}
\usepackage{ulem}

\usepackage{subcaption}

\newtheorem{lemma}{Lemma}
\newtheorem{proposition}{Proposition}

\begin{document}

%*********************************************************************************************************************
%\title{Increasing excess entropy in the approach towards equilibrium in a reversible Ising dynamics model}
%\title{The approach towards equilibrium in a reversible Ising dynamics model -- an exact solution of the one-dimensional Q2R dynamics}
\title{The approach towards equilibrium in a reversible Ising dynamics model -- an information-theoretic analysis\\ based on an exact solution}
\author[1]{Kristian Lindgren\footnote{Corresponding author: kristian.lindgren@chalmers.se}}
\author[2]{Eckehard Olbrich}
%\address{Complex systems group, Department of Energy and Environment, \\ Chalmers University of Technology, SE-41296 G\"oteborg, Sweden\fnref{label3}}
\affil[1]{Complex systems group, Department of Energy and Environment, \authorcr \textit{Chalmers University of Technology, SE-41296 G\"oteborg, Sweden}}
\affil[2]{Max Planck Institute for Mathematics in the Sciences,  \authorcr \textit{Inselstra\ss e 22, 04103 Leipzig, Germany}}

%*********************************************************************************************************************

\date{\today}

\maketitle

\begin{abstract}
We study the approach towards equilibrium in a dynamic Ising model, the Q2R cellular automaton, with microscopic reversibility and conserved energy for an infinite one-dimensional system. Starting from a low-entropy state with positive magnetisation, we investigate how the system approaches equilibrium  characteristics given by statistical mechanics. We show that the magnetisation converges to zero exponentially. The reversibility of the dynamics implies that the entropy density of the microstates is conserved in the time evolution. Still, it appears as if equilibrium, with a higher entropy density is approached. In order to understand this process, we solve the dynamics by formally proving how the information-theoretic characteristics of the microstates develop over time. With this approach we can show that an estimate of the entropy density based on finite length statistics within microstates converges to the equilibrium entropy density. The process behind this apparent entropy increase is a dissipation of correlation information over increasing distances. It is shown that the average information-theoretic correlation length increases linearly in time, being equivalent to a corresponding increase in excess entropy.
\end{abstract}

\section{Introduction}

The apparent contradiction between the reversibility of the microscopic equations of motions and the irreversibility of macroscopic processes has been a problem since the development of statistical mechanics by Maxwell, Boltzmann and Gibbs, see, e.g., refs. \cite{Mackey1989,Lebowitz1999}. How can microscopic reversibility be consistent with macroscopically irreversible phenomena like the second law of thermodynamics?

In this paper we take a microscopic perspective on the development of statistical properties of a system that follows a time evolution that is microscopically reversible. In what way can one understand how such a system "approaches" equilibrium? What is the role of internal correlations of the microstate and how do these change in the time evolution?

As an illustrative model we have chosen the energy conserving Ising dynamics model Q2R in one dimension. We consider the system in the thermodynamic limit, i.e., an infinite sequence of spins, and it is assumed that the initial microstate is generated by a Bernoulli process with a dominating spin direction so that a magnetised and ordered (low entropy) configuration serves as the starting point for the dynamics.

The Q2R rule employs a parallel update according to a checkerboard pattern alternating between the white and black sites.
This leads to a dynamics over a sequence of microstates, with (in general) changing internal statistical properties. Formally, we study how the dynamics changes the stochastic process that characterises the ensemble of microstates at the given time. The initial microstate is spatially ergodic, since it is a Bernoulli process. The same holds for any finite time step, even though a cellular automaton rule in general changes the process so that it becomes a hidden Markov model already after the first iteration.

We characterise the internal disorder (entropy) of a microstate at time $t$ by the entropy density of the corresponding generating process. This entropy is also directly derived from the internal statistics of the microstate by taking into account all possible internal correlations. This can then be viewed as an internal measure of disorder of the microstate -- a microscopic entropy \cite{Lindgren1988}.

Since the dynamics is microscopically reversible, the entropy density is conserved even if the stochastic process that generates the microstates changes \cite{Lindgren1987}. The aim with this paper is to gain a full understanding on how this can be consistent with the apparent picture of a dynamics that brings the magnetised initial state of low entropy density into a state with zero magnetisation and a seemingly higher entropy density.

We solve exactly the dynamics of Q2R in one dimension, starting with a Bernoulli generated microstate, by deriving the statistical properties of the hidden Markov models that generate the microstates at any time $t$.

The picture that emerges is one where some correlations remain at short distance -- in fact, exactly those that make sure that the energy is conserved. It is useful to discuss this in terms of ordered information, or negentropy (as the difference between full disorder and actual entropy density). This ordered information contains information in all correlations in the system, as well as density information, i.e., spin frequencies deviating from $\{1/2,1/2\}$. Except for the nearest neighbour correlations, all other information is transferred to ever increasing distances. This leads to three observations: (i) the magnetisation quickly approaches zero, (ii) the local correlations approach those that characterise an equilibrium microstate at the given energy, (iii) the rest of the correlations (the negentropy) becomes more and more difficult to detect as they require larger and larger blocks of spins and their exact characteristics for their detection.

The focus of the present paper is to examine to what extent this process can be quantified, and whether we can make a more precise statement on how equilibrium is approached on the microscopic level.

%We do this for a microscopically reversible model of Ising dynamics, the Q2R model \cite{Vichniac1984}, that has been extensively used for calculations of equilibrium properties in the two- and three-dimensional Ising model \cite{Stauffer1997}\footnote{Reference has to be checked. Perhaps taking another one? Stauffer \cite{Stauffer2000} cites 8 papers}. In our investigation we instead apply this to the one-dimensional Ising model. Even though the equilibrium state is trivial, the microscopic dynamics that in some sense brings the systems towards equilibrium is less so.

In \cite{Stauffer2000} microscopic reversibility and macroscopic irreversibility for the Q2R automaton was discussed looking at how the period length growth with the system size and thus showing that the recurrence time goes to infinity in the thermodynamic limit. In the present study we want to understand the approach to equilibrium from an information-theoretic point of view.  The aim is to show and quantify how information in correlations are spread out over increasing distances so that, when observing configurations over shorter length scales, it appears as if the system is approaching equilibrium. 

In \cite{Kari2015} it is discussed in what way reversible and, more generally, surjective cellular automata exhibit mixing behaviour in the time evolution, i.e., whether there are cellular automata that in some way can be said to approach a random distribution (Bernoulli distribution with equal probabilities). The most well studied example is the XOR rule, see, e.g., \cite{Lind1984,Lindgren1987}, in which there is a randomization even though there are also recurrent, locally detectable, low entropy states, even for an infinite system. It is stated as an open question whether there are more physically relevant models that allow for a mathematical treatment of how such a mixing may occur, which then would imply a mixing modulo the energy constraint of the system, i.e., a maximization of the entropy density given the energy density \cite{Kari2015}. We contribute to that question by providing the exact solution of the one-dimensional Q2R model as an example of a physically relevant model showing relaxation towards equilibrium.

The plan of the paper is the following: In section \ref{sec:Q2R} we introduce the model system -- the Q2R cellular automaton -- and discuss some of its known properties. In section \ref{sec:time_evolution} we provide the analytical solution for the time evolution of a specific non-equilibrium probability distribution starting from independent spins in the one-dimensional Q2R cellular automaton. 
We use this solution to investigate the time evolution of the information-theoretic quantities and how they are consistent with the system achieving thermodynamic equilibrium. 
In particular we show that the correlation information can be divided into two different contributions -- one part that reflects the equilibrium properties of the system (within interaction distance), and one part with an average correlation length that increases linearly in time.
%In particular we show that the excess entropy can be divided into two terms, one related to the equilibrium properties of the system and a second part linearly increasing in time, which corresponds to the increasing amount if information ... \eoc{necessary for exactly what? Reversing the macro-dynamics?}
In section \ref{sec:Results}, we discuss how the information-theoretic analysis explains how an equilibrium distribution is approached, even though the micro dynamics is reversible. The paper is then concluded by a discussion in section \ref{sec:Discussion}.

\section{Q2R -- a microscopically reversible Ising dynamics}
\label{sec:Q2R}
We consider the Q2R model \cite{Vichniac1984} in one dimension and in the limit of an infinite system. This means that we describe the spatial state (infinite sequence of spins) at a certain time as the outcome of a stationary stochastic process. The system is described as an infinite sequence of states, spin up or spin down, $\uparrow$ and $\downarrow$, respectively. In addition to this a state also holds the information whether to be updated or not in the current time step. The updating structure is such that every second spin is updated at $t$, and then at the next time step the other half of the lattice is updated, and so forth. The updating rule flips a spin when the spin flip does not change the energy, and it changes the state from updating to quiescent and vice versa. Normal nearest neighbour Ising interaction is assumed with an energy $-1$ for parallel spins ($\uparrow\uparrow$ or $\downarrow\downarrow$) and $+1$ for anti-parallel spins ($\uparrow\downarrow$ or $\downarrow\uparrow$). This means that the Q2R model is a micro canonical simulation of the Ising model, with conserved energy. It is also clear that the rule is reversible.

We assume that the initial state is generated by a Bernoulli process, and the aim is to give a statistical analysis of how 
%the spatial configuration changes
spatial configurations change
over time. Each time step is thus characterised by a certain stochastic process, and the Q2R rule transforms this process from one time step to the next.

%\eoc{Should we emphasize that we evolve an ensemble?}

Since the Q2R rule is reversible this implies that the entropy density $h(t)$ of the ensemble at a given time step $t$, or, equivalently, the entropy rate of the stochastic process that generates the ensemble at time step $t$, is a conserved quantity under the Q2R dynamics. This follows, for example, from the observation that there is a local rule (also Q2R, but with a state shift) that runs backwards in time. Since both of these rules are deterministic and thus imply a non-increasing entropy density, the entropy density for an infinite system is conserved under Q2R.

Furthermore, we assume that the stochastic process is ergodic. Note that this is a spatial ergodicity, which implies that for almost all microstates in the ensemble (at any point in time), we have the sufficient statistics to calculate any information-theoretic properties depending on finite length subsystems of the microstate. This means that we can characterise a single microstate, at any time $t$, and identify its internal entropy density and correlation characteristics, which is identical to the same characteristics for the whole ensemble. This is conceptually appealing, since we can then identify an entropy density quantity as a property of a single microstate.

\section{Analysis of the time evolution starting from a Bernoulli distribution}
\label{sec:time_evolution}
We assume that the initial spatial state is described by, or generated by, a Bernoulli process with probability $0<p<1/2$ for spin up. In addition to this we augment our state variable with a second binary variable which marks every second lattice site being in an updating state ($\underline 0$ or $\underline 1$), and the others in quiescent states (0 or 1), where the spin direction is denoted by 0 and 1 (with or without the underline mark) for spin up and down, respectively. Thus the spin up probability is $p(\uparrow)=p(1)+p(\underline 1)$, and similarly for spin down.

The entropy density, i.e., the entropy per site, of such an initial state is simply the entropy of the Bernoulli process,  
\begin{align}
h=S[\{p,1-p\}]:=p \log \frac{1}{p} + (1-p) \log \frac{1}{1-p} \label{h-bernoulli}
\end{align}
since the updating state structure of quiescent and updating states is completely ordered and does not contribute to the entropy density. (The function $S$ is the entropy of the probability distribution, as indicated by the equation.)

With an energy contribution from parallel and anti-parallel spins of $-1$ and $1$, respectively, we get the energy density $u = -(1-2p)^2$ of the initial state. The system is not in equilibrium since the entropy density $h$ is not in a maximum given the energy density $u$. This is obvious already from the fact that the initial magnetisation is positive.

Does the time evolution bring the system closer to the maximum entropy description in some sense, and how? The answer to these questions is the focus of the presented analysis and discussion.

\subsection{Time evolution of the magnetisation}
The Q2R rule in one dimension can be expressed as a simple addition modulo 2 rule for the updating states, $\underline s_{i,t}$, at position $i \in  \mathbb{Z}$ and time $t$, and just a copying of the spin state for the quiescent states, $s_{j,t}$, at $j \in  \mathbb{Z}$, so that at time $t+1$ we get
\begin{align}
{\underline s}_{i,t+1} &= s_{i,t}  \label{q-state-1step}  \\
s_{j,t+1} &= s_{j-1,t} + {\underline s}_{j,t} + s_{j+1,t} \text{     (mod 2)} \;. \label{u-state-1step}
\end{align}
The addition modulo $2$ for the updating states is the operation that flips the spin (0 or 1) whenever that does not change the local energy. This allows us to express the local states at any time, as a sum over initial spin states, using the following Proposition.

\begin{proposition} \label{additive dynamics}
Updating and quiescent states, $\underline s_{i,t}$ and $s_{j,t}$, respectively, at time $t>0$ can be expressed as sums modulo 2 of initial spin states, $\xi_{i'}$ (with $i'\in  \mathbb{Z}$), over certain intervals, 
%of length $2t-1$ and $2t+1$, respectively,
%
\begin{align}
{\underline s}_{i,t} &= \sum_{i'\in I_{i,t-1}} \xi_{i'} \text{     (mod 2)} \;, \label{u-state} \\
s_{j,t} &= \sum_{j'\in I_{j,t}} \xi_{j'} \text{     (mod 2)} \;, \label{q-state}
\end{align}
where $I_{i,n}=\{i-n,...,i,...i+n\}$ denotes the $(2n+1)$-length interval of positions centred around $i$. A local updating state ${\underline s}_{i,t}$ at position $i$ and time $t$ thus depends on $2t-1$ initial stochastic variables, while a quiescent state $s_{j,t}$ depends on $2t+1$ initial stochastic variables.
\end{proposition}

\begin{proof}
We prove this by induction. At time $t=1$ the Q2R rule, Eqs.~(\ref{q-state-1step},~\ref{u-state-1step}), results in ${\underline s}_{i,1} = \xi_i$, and $s_{j,1} = \xi_{j-1}+\xi_{j}+\xi_{j+1}$ (modulo 2). (The addition should here be understood as operating on the spin states, 0 and 1.) Thus Eqs.~(\ref{u-state},~\ref{q-state}) hold for $t=1$. 

If we assume that the Proposition holds for time $t$, then we can use the Q2R rule, Eqs.~(\ref{q-state-1step},~\ref{u-state-1step}), to find the expression for the states at time $t+1$ (where all summations are assumed to be modulo 2),
\begin{align}
{\underline s}_{i,t+1} &= s_{i,t} = \sum_{i'\in I_{i,t}} \xi_{i'} \;,  \\
s_{j,t+1} &=s_{j-1,t}+{\underline s}_{j,t}+s_{j+1,t} = \nonumber \\
&= \sum_{j'\in I_{j-1,t}} \xi_{j'} +  \sum_{j'\in I_{j,t-1}} \xi_{j'} +  \sum_{j'\in I_{j+1,t}} \xi_{j'} = \nonumber \\
&= \xi_{j-t-1}+ \xi_{j-t} + 3 \Big( \sum_{j'\in I_{j,t-1}} \xi_{j'} \Big) + \xi_{j+t}+ \xi_{j+t+1} = \nonumber \\
&=  \sum_{j'\in I_{j,t+1}} \xi_{j'} \;, 
\end{align}
which are Eqs.~(\ref{u-state},~\ref{q-state}) of the Proposition for $t+1$.  \\
\end{proof}

\noindent
This means that we get the distribution for local states, $P_1(t)$, i.e., single site distribution, for $t>0$ determined by 
\begin{align}
p(1,t) = \frac{1}{2} f_\text{odd}^{(2t+1)}(p) \;,\;\;\;\;\;\;p(0,t)=\frac{1}{2}-p(1,t) \;, \label{px} \\ 
p(\underline1,t) = \frac{1}{2} f_\text{odd}^{(2t-1)}(p) \;,\;\;\;\;\;\;p(\underline0,t)=\frac{1}{2}-p(\underline1,t) \;. \label{p_}
\end{align}
Here $f_\text{odd}^{(n)}(p)$ is the probability for getting an odd number of 1's from a process generating $n$ independent symbols (0 or 1) with probability $p$ for each 1, i.e., 
\begin{align}
f_\text{odd}^{(n)}(p) = \sum_{\text{odd }k \in \{0,...,n\}} \binom{n}{k} p^k(1-p)^{n-k} = \frac{1}{2} \big( 1- (1-2p)^n \big) \;. \label{f-odd}
\end{align}
The derivation makes use of $((1-p) - p)^n = \sum_k \binom{n}{k} (-p)^k(1-p)^{n-k} = - f_\text{odd}^{(n)}(p) + (1-f_\text{odd}^{(n)}(p))$. The length of the sequence that affects a spin state at $t$ is $\sim 2t$. An odd number of 1's in the sequence at $t=0$ results in a 1 at the position at time $t$. As $t$ increases, the probability for spin up, $p(\uparrow,t) = p(\underline1,t)+p(1,t) \rightarrow 1/2$. Since the rule transforms an ergodic stochastic process description of configurations at time $t$ to a unique new such process at time $t+1$, this implies that the magnetisation approaches 0. We define the magnetisation by the difference in spin probabilities, i.e., the average of upward spins (with direction $+1$) and downward spins (with direction $-1$),
\begin{equation}
m(t) = p(\uparrow,t) - p(\downarrow,t) = 2p(\uparrow,t)-1 \;. 
\end{equation} 
The approach to zero magnetisation is then given by the following proposition.

\begin{proposition} \label{magnetisation}
The magnetisation, $m(t)$, is given by
\begin{equation}
m(t) = - \frac{1}{2} \big( (1-2p)^{2t-1} + (1-2p)^{2t+1} \big) \;. 
\end{equation} 
\end{proposition}

\begin{proof}
The frequency of state $\uparrow$ at time $t$ is given by two binomial distributions (for updating and quiescent states, respectively) and their corresponding probabilities for having an odd number of 1's,
\begin{equation}
p(\uparrow,t) = f_\text{odd}^{(2t-1)}(p) + f_\text{odd}^{(2t+1)}(p)  = \frac{1}{2} - \frac{1}{4} (1-2p)^{2t-1} - \frac{1}{4} (1-2p)^{2t+1} \;,
\end{equation} 
where we have used Eq.~(\ref{f-odd}). This proves the Proposition.  
\end{proof}

\noindent
Thus we have an exponential convergence towards zero magnetisation. The frequency of spin up (and down) quickly approaches 1/2. For example, with an initial frequency of $p=0.2$, we have after $t=10$ time steps, $p(\uparrow,t)=0.499979...$ .

\subsection{Time evolution of information-theoretic characteristics}
In order to analyse how the Q2R dynamics transform the initial state (distribution of states) to states that in some way resembles equilibrium states, we make an information-theoretic analysis of the spatial configurations at the different time steps $t$, i.e., the stochastic processes that characterise those configurations.

\subsubsection{Information theory for symbol sequences}
The basis for the information-theoretic formalism is the set of probability distributions for $m$-length sequences at time $t$, determined by the corresponding stochastic process characterising the spatial configuration, $P_m(t) = \{p(x_1,...,x_m)\}_{x_i \in \{0,1,\underline0,\underline1 \}}$. 
All the key quantities for characterising order and disorder can be expressed in terms of block entropies, $H_m(t)$,
\begin{align}
H_m(t) = S[P_m(t)]=-\sum_{x_1,...,x_m} p(x_1,...,x_m) \log p(x_1,...,x_m)  \label{H-block}\;.
\end{align}
The entropy density 
\begin{align}
 h(t)=h=\lim_{m \to \infty} \frac{H_m(t)}{m} \;, \label{entropy-density}
\end{align}
expressing the randomness of the stochastic process generating microstates at time $t$, is conserved since the dynamics is reversible. The conditional entropy 
\begin{align}
 h_m(t) =-\sum_{x_1,...,x_m} p(x_1,...,x_m) \log p(x_m|x_1,...,x_{m-1})  \label{cond-h}\;.
\end{align}
can be expressed as the difference of two block entropies 
\begin{align}
h_m(t) = H_m(t) - H_{m-1}(t) \;,
\end{align}
and can be interpreted as a finite length estimate of the entropy rate. The full entropy density is recovered in the infinite length limit, because of the spatial stationarity. The entropy density $h$ can thus be expressed as
\begin{align}
h = \lim_{m \rightarrow \infty} h_m(t)  \;.
\end{align}
The decrease in the estimate of the entropy density $h_m(t)$ as $m$ increases quantifies correlation information $k_m(t)$ in blocks over length $m$, 
\begin{align}
k_m(t) &= h_{m-1}(t)-h_m(t) = \nonumber \\[4pt]
&= -H_{m} + 2H_{m-1} - H_{m-2} \ge 0  \;. \label{k_m}
\end{align}
Here we define, $k_2(t)=-H_2(t)+2H_1(t)$ representing neighbouring pair correlation information (which is equal to the mutual information \cite{Cover1991}), and $k_1(t)=\log 4-H_1(t)$ representing density information. It is sometimes useful to work with the Kullback-Leibler form of the correlation information, based on conditional probabilities \cite{Eriksson1987,Lindgren1987},
\begin{align}
k_m(t) = \sum_{x_1,...,x_{m-1}} p(x_1,...,x_{m-1}) \sum_{x_{m}} p(x_m | x_1,...,x_{m-1}) \log \frac {p(x_m | x_1,...,x_{m-1})}{p(x_m | x_2,...,x_{m-1})} \ge 0  \;.\label{k_m-2}
\end{align}
Note that $k_m(t)$ is also the conditional mutual information between $x_1$ and $x_m$ given $x_2,\ldots,x_{m-1}$ which quantifies the additional amount of information that $x_1$ provides about the $m$th symbol, on average, given that one knows already $x_2,\ldots,x_{m-1}$.  
With $n=4$ possible states per lattice site, the total information of $\log n=\log 4$ can be fully decomposed into the introduced quantities,
\begin{align}
\log 4 = h + k_\text{corr} = h + \sum_{m=1}^\infty k_m(t)  \;, \label{info-decmoposition}
\end{align}
where we have introduced the total correlation information, $k_\text{corr}$, as the sum over all contributions $k_m(t)$, including the density information $k_1(t)$. 
This means that the estimate of the entropy density $h_m(t)$, based on blocks of length up to $m$, can be written
\begin{align}
h_m(t) = \log 4 - \sum_{j=1}^m k_j(t)  \;. \label{h_m}
\end{align}
The most common information-theoretic quantity for characterising complexity is the Excess entropy \cite{Shaw1984,Crutchfield2003} or Effective measure complexity \cite{Grassberger1986}, $\eta(t)$,
\begin{align}
\eta(t) =  \lim_{m \rightarrow \infty} H_m(t) - m \, h  \;. \label{eta-def}
\end{align}
The excess entropy can be expressed as a weighted sum over correlation information terms,
\begin{align}
\eta(t) =  \sum_{m=2}^\infty (m-1) k_m(t)  \;, \label{eta}
\end{align}
and thus it reflects an average information-theoretic correlation length \cite{Lindgren1987}. Since the entropy density $h$ is conserved in the time evolution, so is also the correlation information $k_\text{corr}$. But the lengths at which correlation information is located may change over time, and thus we would in general expect the excess entropy, or the average correlation length, to change over time.

\subsubsection{Some special properties for the Q2R model in one dimension}

Let $P^{(n)}$ be the distribution over an odd or an even number of $1$'s in a sequence of $n$ independently generated $0$'s or $1$'s (with probability $p$ for $1$), i.e.,
\begin{align}
P^{(n)} = \Big\{ f_\text{odd}^{(n)}(p), 1-f_\text{odd}^{(n)}(p) \Big\}  \;,\label{P-odd}
\end{align}
with $f_\text{odd}^{(n)}$ defined by Eq.~(\ref{f-odd}).\\ 

\noindent
Because we have a system with a strong periodic order of updating and quiescent states, it is convenient to split the correlation information quantities into two different types, $k_{m}^{(-)}$ and $k_{m}^{(\times)}$, respectively, depending on whether the last state in the block over which correlations are considered is updating or quiescent ($\underline s_m$ or $s_m$), 
\begin{align}
%k_{m} = k_{m-} + k_{m\times} \;. \\
k_{m} = k_{m}^{(-)} + k_{m}^{(\times)} \;. \label{k_m-decompose}
\end{align}
We introduce a corresponding notation for the block entropies, $H_{m}^{(-)}$ and $H_{m}^{(\times)}$, referring to blocks ending with an updating and a quiescent state, respectively. This means that $H_{m}^{(-)}$ is defined by  
\begin{align}
H_{m}^{(-)} = -\sum_{x_1,...,x_{m-1}} \sum_{x_m \in \{\underline0, \underline1\} } p(x_1,...,x_m) \log p(x_1,...,x_m)  \label{Hq-block}\;,
\end{align}
which means that every second $x_i$ to the left of $x_m$ also needs to be an updating state for the corresponding term to contribute. The block entropy $H_{m}^{(\times)}$ is defined similarly. Then $k_{m}^{(-)}$ and $k_{m}^{(\times)}$ can be written
\begin{align}
k_{m}^{(-)} &= -H_{m}^{(-)} + H_{m-1}^{(-)}+H_{m-1}^{(\times)} - H_{m-2}^{(\times)} \;, \label{km_} \\
k_{m}^{(\times)} &= -H_{m}^{(\times)} + H_{m-1}^{(\times)}+H_{m-1}^{(-)} - H_{m-2}^{(-)} \;, \label{km+}
\end{align}
where we have used Eqs.~(\ref{H-block},~\ref{k_m},~\ref{k_m-2},~\ref{Hq-block}). We observe that, for a spatially symmetric system,
\begin{align}
H_{2m}^{(-)} = H_{2m}^{(\times)}  \;, \label{H2m}
\end{align}
since an even length sequence that ends with an updating state needs to start with a quiescent state. 

\subsubsection{Correlation characteristics of the Q2R model in one dimension}

The following propositions assume that the initial state at time $t=0$ is generated as described above: The spin states are generated by a Bernoulli process with probability $p$ for spin up. Then the alternating order of updating and quiescent states are added on top of this. The entropy density at any $t$ is then determined by the initial entropy density, i.e., $h=-p \log p - (1-p) \log(1-p)$. The goal is to derive expressions that describe how the different contributions to the correlation information may change over time, and how that affects the estimates of the entropy density $h_m(t)$. \\
%a closed expression for the excess entropy as a function of time as a characterisation of the information dynamics in correlations.\\

\begin{proposition} \label{Prop:k_2m}
\noindent
%{\bfseries Proposition 1:} \\
There is no correlation information over even length blocks larger than 2. For $m \ge 2$,
\begin{align}
k_{2m}(t)=0 \;. 
\end{align}
\end{proposition}
\begin{proof}
We start with the following observation.\\
{\em Observation:}
\begin{align}
k_{m}^{(\times)}(t) = 0 \text{   for $m>2$} \;. \label{q_no-corr} 
\end{align}
This follows from the fact that, at any $t$, the conditional probability of a quiescent state $s_{m,t}$ does not change when adding new information in states beyond (to the left of) $\underline s_{m-1,t}$. From Eqs.~(\ref{u-state},~\ref{q-state}) we see that $s_{m,t}=\xi_{m-t}+...+\xi_{m+t}$ (mod 2), while $\underline s_{m-1,t}=\xi_{m-t}+...+\xi_{m+t-2}$ (mod 2), i.e., 
\begin{align}
s_{m,t}=\underline s_{m-1,t}+\xi_{m+t-1}+\xi_{m+t} \; \text{  (mod 2).}  \label{state_x}
\end{align}
But the two stochastic variables $\xi_{m+t-1}$ and $\xi_{m+t}$ are not part of any of the states $s_{m',t}$ further left (with position $m'<m-1$). This means that $p(s_{m,t} | \,...\, s_{m-2,t}, \underline s_{m-1,t})=p(s_{m,t} | \underline s_{m-1,t})$, which with Eq.~(\ref{k_m-2}) proves the observation.

Thus we know that, for $m\ge2$, $k_{2m}^{(\times)}(t)=0$. From Eqs.~(\ref{km_},~\ref{km+},~\ref{H2m}) we see that $k_{2m}^{(-)}(t)=k_{2m}^{(\times)}(t)=0$, and the Proposition then follows from Eq.~(\ref{k_m-decompose}).   \\
\end{proof}

\begin{proposition} \label{Prop:k_2m+1}
A correlation information quantity at distance $2m-1$, for $m>1$, is transferred to a correlation information quantity at distance $2m+1$ in the next time step. Thus, for $t>0$ and $ m > 1$,
\begin{align}
k_{2m+1}(t)=k_{2m-1}(t-1) \;. 
\end{align}
\end{proposition}

\begin{proof}
Since, from Eq.~(\ref{q_no-corr}), $k_{2m+1}^{(\times)}(t)=0$, we need to show that $k_{2m+1}^{(-)}(t)=k_{2m-1}^{(-)}(t-1)$. We start with the observation that, for $m\ge3$,
\begin{align}
H_{m}^{(\times)}=H_{m-1}^{(-)} + H_{2}^{(\times)} - H_{1}^{(-)} \;, \label{Hmx-split}
\end{align}
which follows from the fact that the probability in the block entropy can be written, in the case of even $m$, $p(\underline s_1,...,\underline s_{m-1}, s_m) = p(\underline s_1,...,\underline s_{m-1}) p(s_m | \underline s_1,...,\underline s_{m-1})$. From the previous observation, Eq.~(\ref{q_no-corr}), we find that $p(\underline s_1,...,\underline s_{m-1}, s_m) = p(\underline s_1,...,\underline s_{m-1})\, p(\underline s_{m-1},s_m)\, /\, p(\underline s_{m-1})$, and this results in the entropy above. The same argument goes for odd $m$.

We now rewrite the correlation information $k_{2m+1}^{(-)}(t)$ in terms of the block entropies,
\begin{align}
k_{2m+1}^{(-)}(t) &= -H_{2m+1}^{(-)} + H_{2m}^{(-)}+H_{2m}^{(\times)} - H_{2m-1}^{(\times)} = \nonumber \\
&=-H_{2m+1}^{(-)} +2 H_{2m}^{(\times)} - \big(H_{2m-2}^{(-)} + H_{2}^{(\times)} - H_{1}^{(-)}\big) = \nonumber \\
&=-H_{2m+1}^{(-)} +2 \big(H_{2m-1}^{(-)}+ H_{2}^{(\times)} - H_{1}^{(-)}\big) - \big(H_{2m-3}^{(-)} +2\big(H_{2}^{(\times)} - H_{1}^{(-)}\big)\big)= \nonumber  \\
&=-H_{2m+1}^{(-)} +2 H_{2m-1}^{(-)} - H_{2m-3}^{(-)}  \;, \label{k-at-t}
\end{align}
where we have used Eqs.~(\ref{H2m}, \ref{Hmx-split}).

The block entropies are now in a form that makes it possible to transfer the expression to a corresponding one for the previous time step $t-1$. We denote the block entropies at that time step with $H_{m}'$. The symbol sequences considered in Eq.~(\ref{k-at-t}) are all of odd length and of the form $(\underline s_1, s_2,...,s_{2m}, \underline s_{2m+1})$. Since the end states in the sequence are quiescent states at $t-1$, they are just copied, and we note that there is a one-to-one mapping from the sequence at time $t$ to a corresponding $(2m+1)$-length sequence at time $t-1$: $(s_1, \underline s_2',...,\underline s_{2m}', s_{2m+1})$, but with the opposite arrangement of updating and quiescent states. The updating states at time $t-1$ may of course have different spin states. The one-to-one mapping, though, results in that the corresponding block entropies are the same, i.e., $H_{2m+1}^{(-)}=H_{2m+1}'^{(\times)}$. The correlation information $k_{2m+1}^{(-)}(t)$ can then be expressed in terms of block entropies at $t-1$,
\begin{align}
k_{2m+1}^{(-)}(t) &= -H_{2m+1}'^{(\times)} +2 H_{2m-1}'^{(\times)} - H_{2m-3}'^{(\times)}  \;.
\end{align}
As above, we can reduce the block length of $H^{(\times)}$ entropies,
\begin{align}
k_{2m+1}^{(-)}(t) &= -H_{2m}'^{(-)} - H_{2}'^{(\times)} + H_{1}'^{(-)}  +2 \big(H_{2m-2}'^{(-)} + H_{2}'^{(\times)} - H_{1}'^{(-)} \big) - H_{2m-3}'^{(\times)}  = \nonumber \\
&= -H_{2m}'^{(\times)} +2 H_{2m-2}'^{(-)} - H_{2m-3}'^{(\times)} + H_{2}'^{(\times)} - H_{1}'^{(-)} = \nonumber \\
&= -H_{2m-1}'^{(-)} +2 H_{2m-2}'^{(-)} - H_{2m-3}'^{(\times)}  = \nonumber \\
&= k_{2m-1}^{(-)}(t-1) \;.
\end{align}
where we have used Eqs.~(\ref{km_}, \ref{H2m}, \ref{Hmx-split}). This concludes the proof.  \\
\end{proof}

\begin{lemma} \label{Lemma:H_1}
The entropy $H_1(t)$ of a single site is given by:
\begin{align}
H_1(0) &= S[P^{(1)}] + \log 2 = h + \log 2 \;,  \\
H_1(t) &= \frac{1}{2} S[P^{(2t+1)}] + \frac{1}{2} S[P^{(2t-1)}] + \log 2 \;   \text{  ,     for $t>0$ .}
\end{align}
\end{lemma}

\begin{proof}
The single site distribution is given by the $\{1/2, 1/2\}$ distribution for updating and quiescent states multiplied with the distribution giving the probability for spin state 0 or 1 for the updating and quiescent states, $P^{(2t-1)}=\{f_\text{odd}^{(2t-1)}(p), 1-f_\text{odd}^{(2t-1)}(p) \}$ and $P^{(2t+1)}=\{f_\text{odd}^{(2t+1)}(p), 1-f_\text{odd}^{(2t+1)}(p) \}$ , respectively, in the case when $t>0$. For $t=0$, we instead have $P^{(1)}=\{p,1-p\}$ for the spin state. This directly results in the Lemma.  \\
\end{proof}

\begin{lemma} \label{Lemma:H_2}
The 2-length block entropy $H_2(t)$ is given by:
\begin{align}
H_2(0) &= 2S[P^{(1)}] + \log 2 = 2 h + \log 2 \;, \\
H_2(t) &= S[P^{(2)}] + S[P^{(2t-1)}] + \log 2 \;\text{  ,     for $t>0$ .}
\end{align}
\end{lemma}

\begin{proof}
For $t=0$, we directly get the result of Lemma \ref{Lemma:H_2}
\begin{align}
H_2(0) &= 2 h + \log 2  \;,
\end{align}
with the $2h$ from the Bernoulli process and the $\log 2$ from the periodic structure, i.e., the two updating/quiescent possibilities $(-, \times)$ and $(\times, -)$ for pairs.

We can express probabilities for a pair of adjacent states, $(\underline s_{i-1}, s_i)$, at time $t$, by $p_t(\underline s_{i-1}, s_i) = p_t(s_i | \underline s_{i-1})\;p_t(\underline s_{i-1})$. By symmetry, the other order, $(s_{i-1}, \underline s_i)$, gives the same contribution to the block entropy. The resulting entropy can then be expressed as the sum of a conditional entropy and a single cell entropy, 
\begin{align}
H_2(t) &= 2 \Big( \sum_{\underline s_{i-1}} p_t(\underline s_{i-1}) \sum_{s_{i}} p_t(s_i | \underline s_{i-1}) \log \frac{1}{p_t(s_i | \underline s_{i-1})} + \sum_{\underline s_{i-1}} p_t(\underline s_{i-1}) \log \frac{1}{p_t(\underline s_{i-1})} \Big)  \;.
\end{align}
The last sum uses $p_t(\underline s)$ being a non-normalized distribution $\{\frac{1}{2} f_\text{odd}^{(2t-1)},\frac{1}{2}-\frac{1}{2} f_\text{odd}^{(2t-1)} \}$, from Eq.~(\ref{p_}), resulting in $\frac{1}{2}(\log 2 + S[P^{(2t-1)}])$. 
The conditional probabilities of the first sum are derived from Eq.~(\ref{state_x}), showing that $p_t(s_i | \underline s_{i-1})$ is a distribution determined by the sum of two independent stochastic variables, which results in an entropy $S[P^{(2)}]$, regardless of $\underline s_{i-1}$. The sum over $\underline s_{i-1}$ then gives a factor of $\frac{1}{2}$.  This results in
\begin{align}
H_2(t) &= 2 \Big( \frac{1}{2} S[P^{(2)}] + \frac{1}{2} \log 2 + \frac{1}{2} S[P^{(2t-1)}] \Big)  \;,
\end{align}
which gives us the Lemma.  \\
\end{proof}

\begin{proposition} \label{Prop:k_2}
The pair correlation between neighbour states $k_2(t)$ is given by
\begin{align}
k_2(0) &= \log 2 \;, \\
k_2(t) &=  S[P^{(2t+1)}] -  S[P^{(2)}] + \log 2 \; \text{  ,     for $t>0$ .} \label{k_2(t)}
\end{align}
\end{proposition}

\begin{proof}
For $t=0$, we only have correlation from the alternating updating and quiescent cells structure, i.e., $\log 2$. For $t>0$ we use the definition, $k_2 = -H_2 + 2H_1$. Lemmas \ref{Lemma:H_1} and \ref{Lemma:H_2} then immediately result in the Proposition.   \\
\end{proof}

\begin{proposition} \label{Prop:k_3}
The correlation $k_3(t)$ over blocks of length 3 is given by,
\begin{align}
k_3(0) &=0 \;, \\
k_3(1) &= - \frac{1}{2} S[P^{(3)}] + S[P^{(2)}] - \frac{1}{2} S[P^{(1)}] \;, \\
k_3(t) &=  - \frac{1}{2} S[P^{(2t+1)}] + S[P^{(2t-1)}] - \frac{1}{2} S[P^{(2t-3)}]  \; \text{  ,     for $t>1$ .}
\end{align}
\end{proposition}

\begin{proof}
For $t=0$, there is obviously no correlation over blocks larger than 2. 

For larger time steps, we use the fact that in total all correlation information is conserved in the time evolution (since the entropy density $h$ is). Propositions \ref{Prop:k_2m} and \ref{Prop:k_2m+1} imply that all correlation information from blocks of length $m$, with $m \ge 3$, is transferred to longer distances, $m+2$, in the next time step. Therefore, the correlation of blocks over length 3, $k_3(t)$, must come from the change in $k_1 + k_2$ in the last time step, so that $k_\text{corr}$ is conserved,
\begin{align}
k_3(t) &=  k_1(t-1)+k_2(t-1)-k_1(t)-k_2(t) \;.
\end{align}
For $t=1$, Lemma 1 and Proposition \ref{Prop:k_2}
%, with Eqs.~(), 
results in
\begin{align}
k_3(1) &=  \log 4 - H_1(0) + \log 2 - \big(\log 4 - H_1(1) + S[P^{(3)}] -  S[P^{(2)}] + \log 2 \big)  =  \nonumber \\ 
&=  - S[P^{(1)}] - \log 2 + \frac{1}{2} S[P^{(3)}] + \frac{1}{2} S[P^{(1)}] + \log 2 - S[P^{(3)}] +  S[P^{(2)}]  =  \nonumber \\ 
&= -\frac{1}{2} S[P^{(3)}] + S[P^{(2)}] - \frac{1}{2} S[P^{(1)}] \;.
\end{align}
Finally, for $t>1$, Lemma 1 and Proposition \ref{Prop:k_2}, gives us,
\begin{align}
k_3(t) &=  S[P^{(2t-1)}] - S[P^{(2t+1)}] + H_1(t) - H_1(t-1) = \nonumber \\
&= - \frac{1}{2} S[P^{(2t+1)}] + S[P^{(2t-1)}] - \frac{1}{2} S[P^{(2t-3)}]   \;,
\end{align}
which concludes the proof.  \\
\end{proof}

\noindent
We note that Propositions \ref{Prop:k_2m}-\ref{Prop:k_3} imply that the state at time $t$ is characterized by a maximum information-theoretic correlation length of $2t+1$, and thus a Markov process of memory $2t$. Since we now have closed form expressions for all correlation information contributions, we can derive a closed form expression also for the excess entropy.   \\

\begin{proposition} \label{Prop:eta-new}
The excess entropy, $\eta(t)$, is linearly increasing in time after the first time step.
\begin{align}
\eta(0) &= \log 2 \;, \\
\eta(t) &= (2t - 1) \zeta + \log 2 \;, \text{ for $t>0$, } \label{Prop-eta}
\end{align}
where the constant $\zeta$ is determined by
\begin{align}
\zeta &=  S[P^{(2)}]-S[P^{(1)}] \; . \label{zeta} 
\end{align}
\end{proposition}

\begin{proof}
We use the form of $\eta$ which is a weighted sum over correlation information contributions, Eq.~(\ref{eta}), i.e., $\eta=\sum_m (m-1)k_m$. For $t=0$ we only have contribution from $k_2(0)=\log 2$, which gives $\eta(0)=\log 2$.

For $t=1$, Propositions \ref{Prop:k_2} and \ref{Prop:k_3} result in
\begin{align}
\eta(1) &= k_2(1) + 2k_3(1) = S[P^{(3)}] -  S[P^{(2)}] + \log 2 + 2 \big( - \frac{1}{2} S[P^{(3)}] + S[P^{(2)}] - \frac{1}{2} S[P^{(1)}] \big) = \nonumber \\
&=S[P^{(2)}] - S[P^{(1)}] + \log 2 = \zeta + \log 2  \;. \label{eta(1)}
\end{align}
For $t\ge 1$, we note from Propositions \ref{Prop:k_2m} and \ref{Prop:k_2m+1} that the excess entropy, using Eq.~(\ref{eta}), is a weighted sum over $k_2(t)$ and all $k_3(t')$ up to the current time step $(1\le t' \le t)$,
\begin{align}
\eta(t) &= k_2(t) + 2 k_3(t) + 4 k_5(t) + 6 k_7(t) + ... + 2t k_{2t+1}(t) = \nonumber \\
&= k_2(t) + 2 \sum_{t' = 1}^t (t-t'+1)k_3(t')  \;,
\end{align}
where we have used the result from Proposition \ref{Prop:k_2m+1} that $k_{2m+1}(t)=k_3(t-m+1)$. We can then derive an expression for the change of $\eta$ in one time step (for $t\ge 1$),
\begin{align}
\eta(t+1)-\eta(t) &= k_2(t+1)-k_2(t) + 2 \sum_{t' = 1}^{t+1} k_3(t')  \;.
\end{align}
By using Propositions \ref{Prop:k_2} and \ref{Prop:k_3}, we find that this sum results in
\begin{align}
\eta(t+1)-\eta(t) &=  2\big(S[P^{(2)}]-S[P^{(1)}] \big) = 2 \zeta \;.
\end{align}
In combination with Eq.~(\ref{eta(1)}), we then get Eq.~(\ref{Prop-eta}) of the Proposition. 
This concludes the proof.  \\
\end{proof}

\noindent
This result shows that we have an average correlation length (in the information-theoretic sense), cf. Eq.~(\ref{eta}), that increases linearly in time. The immediate implication is that some of the information that was initially detectable by looking at statistics over subsequences of shorter lengths is not any longer found at those length scales. 
How this relates to the approach towards a distribution closer to an equilibrium one is discussed in the following section.

%\subsection{Evolution towards local distributions indistinguishable from the equilibrium distribution}
\section{The microscopically reversible approach towards the equilibrium distribution}
\label{sec:Results}
The equilibrium distribution for a one-dimensional system can be derived by finding the distributions $P_m$, over $m$-length blocks, that maximise the entropy density $h$, Eq.~(\ref{entropy-density}), under an energy constraint. Here the initial density $p$ of spin up implies an energy density per site $u=-(1-2p)^2$.

Because of the dual relation between entropy density $h$ and correlation information $k_\text{corr}$, as is seen in Eq.~(\ref{info-decmoposition}), the equilibrium can also be determined by minimisation of $k_m$. With nearest neighbour interaction only, we can always choose $k_m=0$ for $m>2$, since the constraint need not give rise to higher order correlations \cite{Lindgren1987}. As is well known, this results in magnetisation 0, or $p(\uparrow)=p(\downarrow)=1/2$. With a normalization constraint this implies that $p(\uparrow \uparrow)=p(\downarrow \downarrow)=1/2(1-2p(\uparrow \downarrow))$, which with the energy constraint fully determines the distribution over pairs of spins. The solution is simply that $p(\uparrow \downarrow)=p(1-p)$, i.e., the same as the initial one determined by the Bernoulli process. This follows from the fact that if energy is to be conserved, then $p(\uparrow \uparrow)+p(\downarrow \downarrow)$ as well as $p(\uparrow \downarrow)$ must be conserved. 

So, in equilibrium, we have that $p(\uparrow \downarrow)=p(1-p)$ and $p(\uparrow \uparrow)=(p^2+(1-p)^2)/2$, and the same for the corresponding symmetric configurations. This distribution then determines the $2$-block entropy $H_{2,\text{eq}}$ and a single state entropy $H_{1,\text{eq}}=\log 2$. This results in an equilibrium entropy density, $h_\text{eq}$,
\begin{align}
h_\text{eq} &= H_{2,\text{eq}} - H_{1,\text{eq}} =\nonumber \\ 
&= -2p(1-p) \log (p(1-p)) - (p^2+(1-p)^2)\log((p^2+(1-p)^2)/2) - \log 2 = \nonumber \\
&= S[P^{(2)}] \;,
\end{align}
where we have used Eq.~(\ref{P-odd}). 
Note that the entropy density of the studied Q2R system is $h=S[\{p,1-p\}]=S[P^{(1)}]<S[P^{(2)}]=h_\text{eq}$.

\begin{figure}[h]
        \centering
        \begin{subfigure}[b]{0.48\textwidth}
                \includegraphics[width=\textwidth]{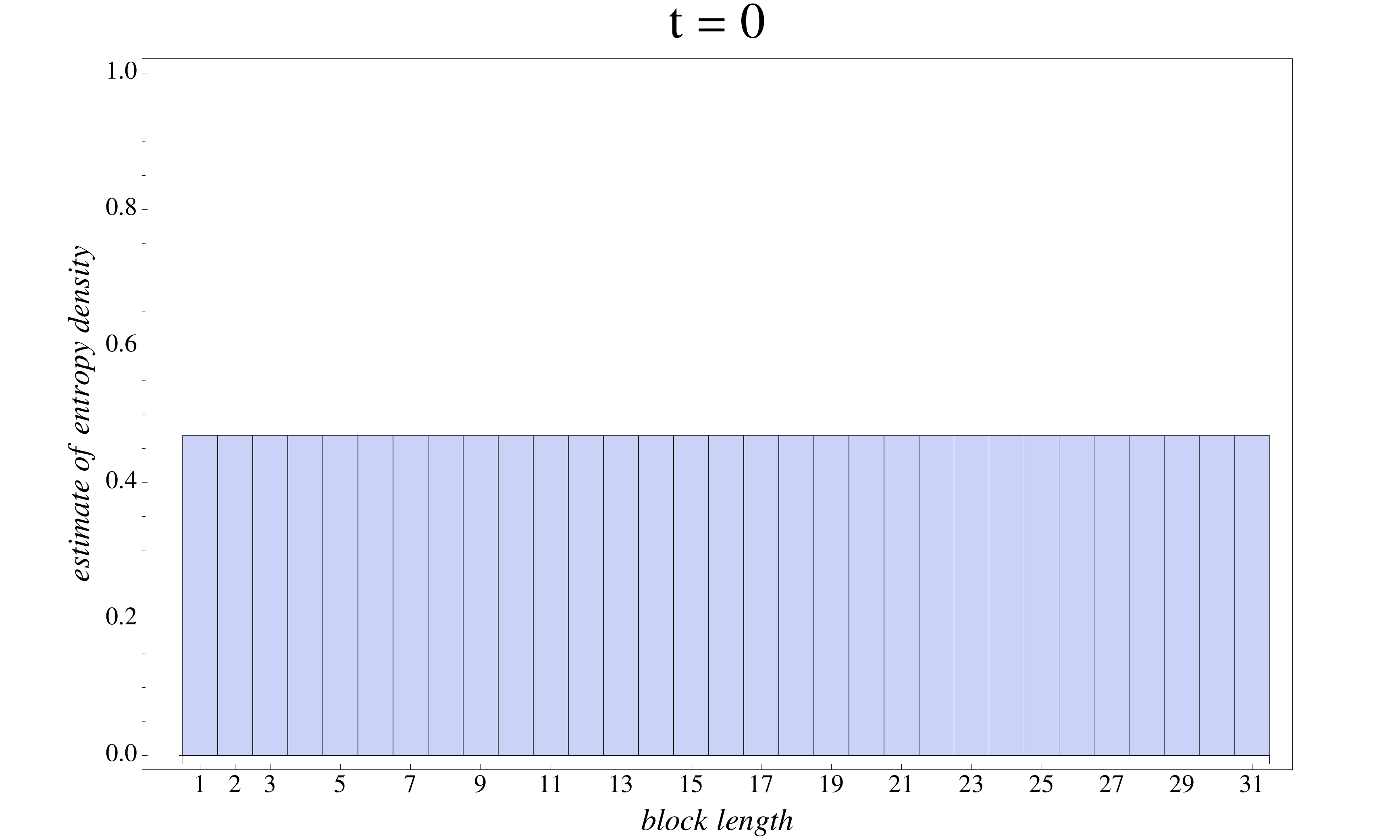}
                \caption{}
                \label{fig:st0}
        \end{subfigure}
        \begin{subfigure}[b]{0.48\textwidth}
                \includegraphics[width=\textwidth]{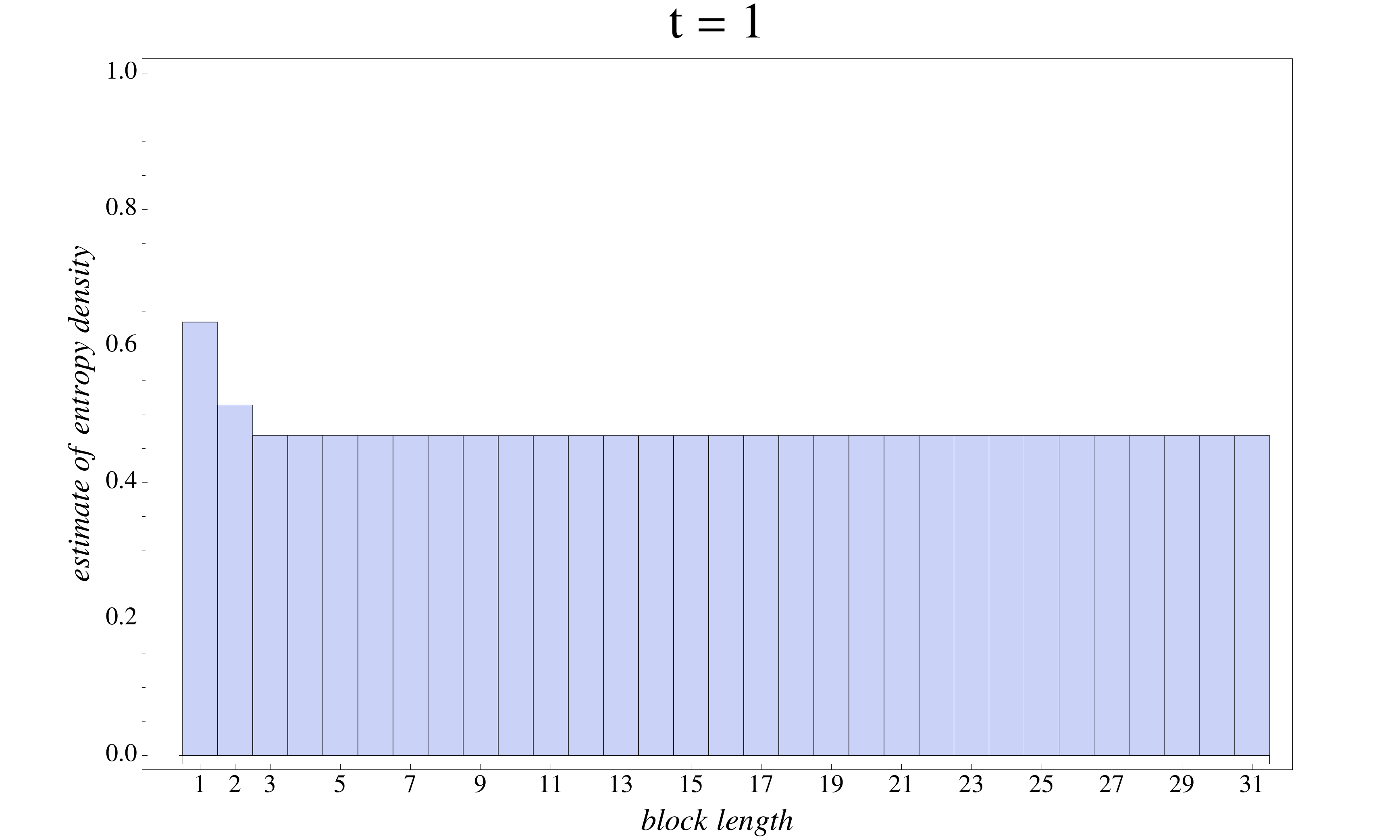}
                \caption{}
                \label{fig:st1}
        \end{subfigure}
        \\[12pt]
        \begin{subfigure}[b]{0.48\textwidth}
                \includegraphics[width=\textwidth]{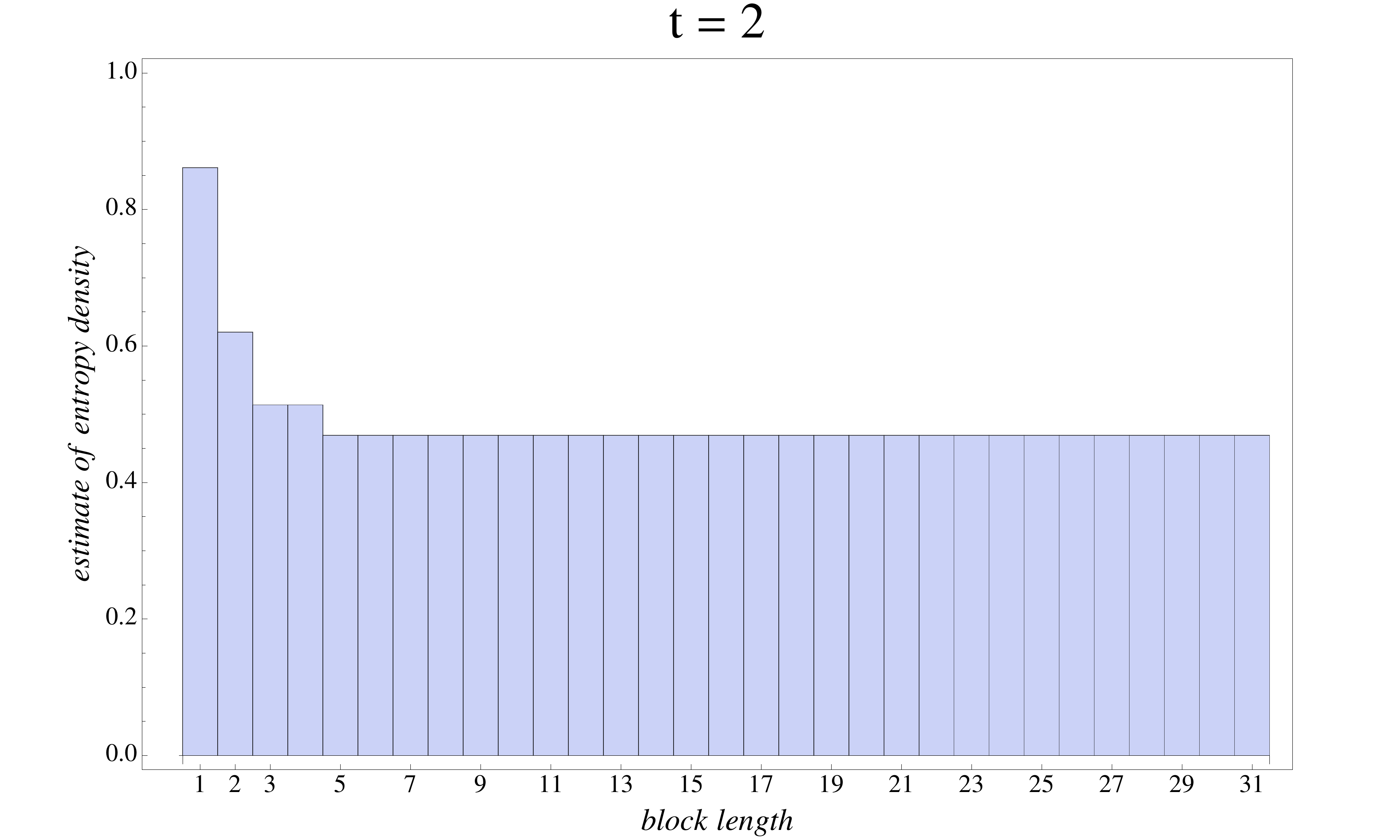}
                \caption{}
                \label{fig:st2}
        \end{subfigure}
        \begin{subfigure}[b]{0.48\textwidth}
                \includegraphics[width=\textwidth]{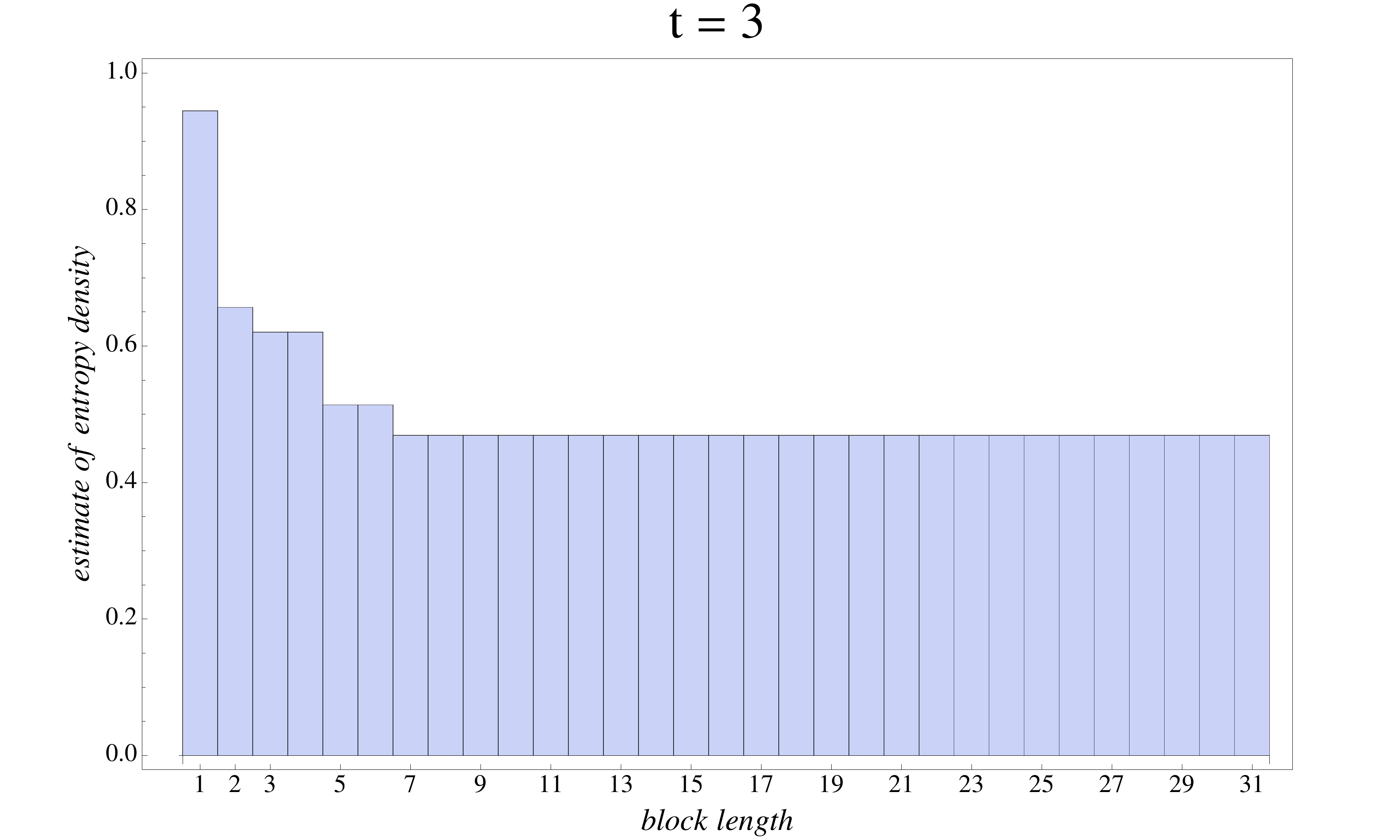}
                \caption{}
                \label{fig:st3}
        \end{subfigure}
        \\[12pt]
        \begin{subfigure}[b]{0.48\textwidth}
                \includegraphics[width=\textwidth]{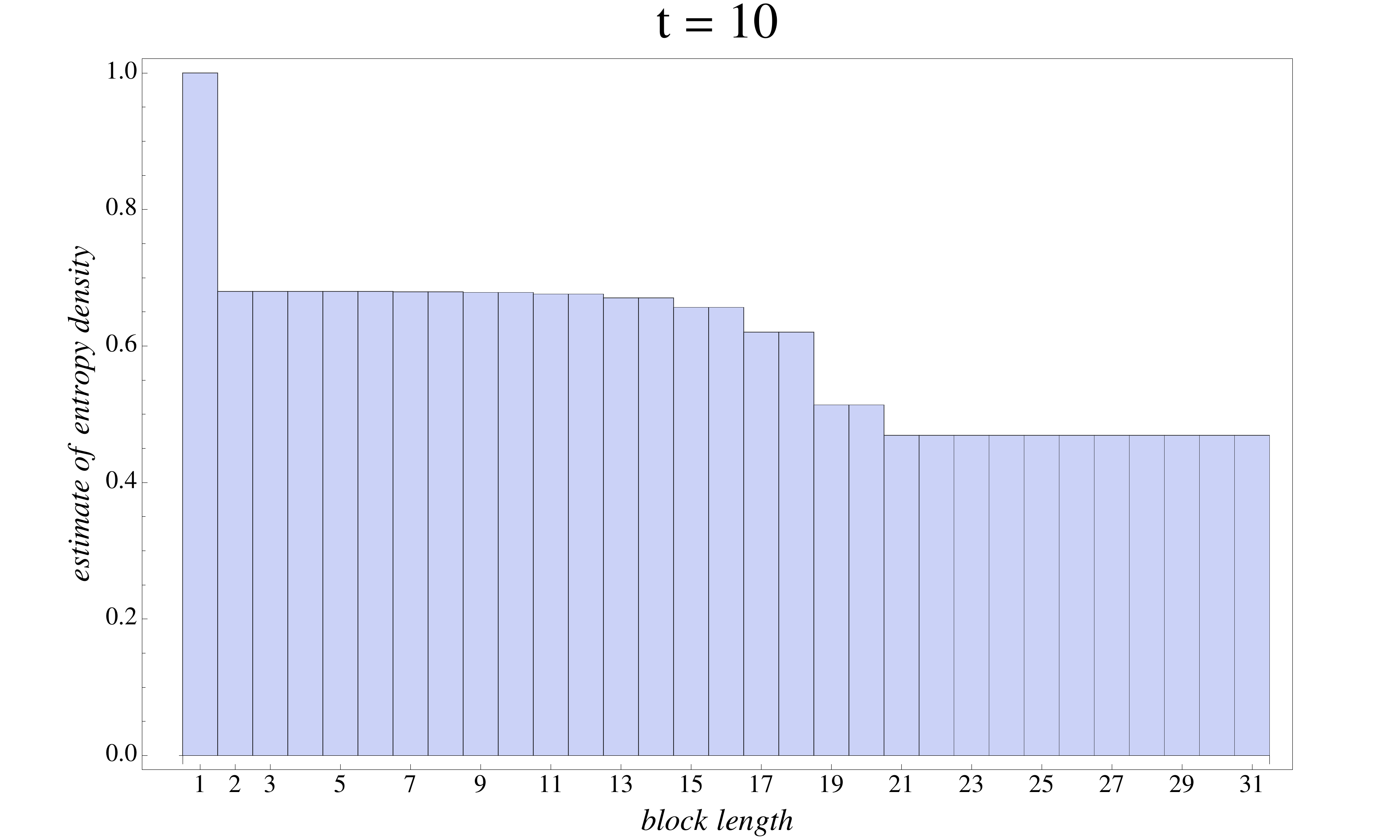}
                \caption{}
                \label{fig:st10}
        \end{subfigure}
        \begin{subfigure}[b]{0.48\textwidth}
                \includegraphics[width=\textwidth]{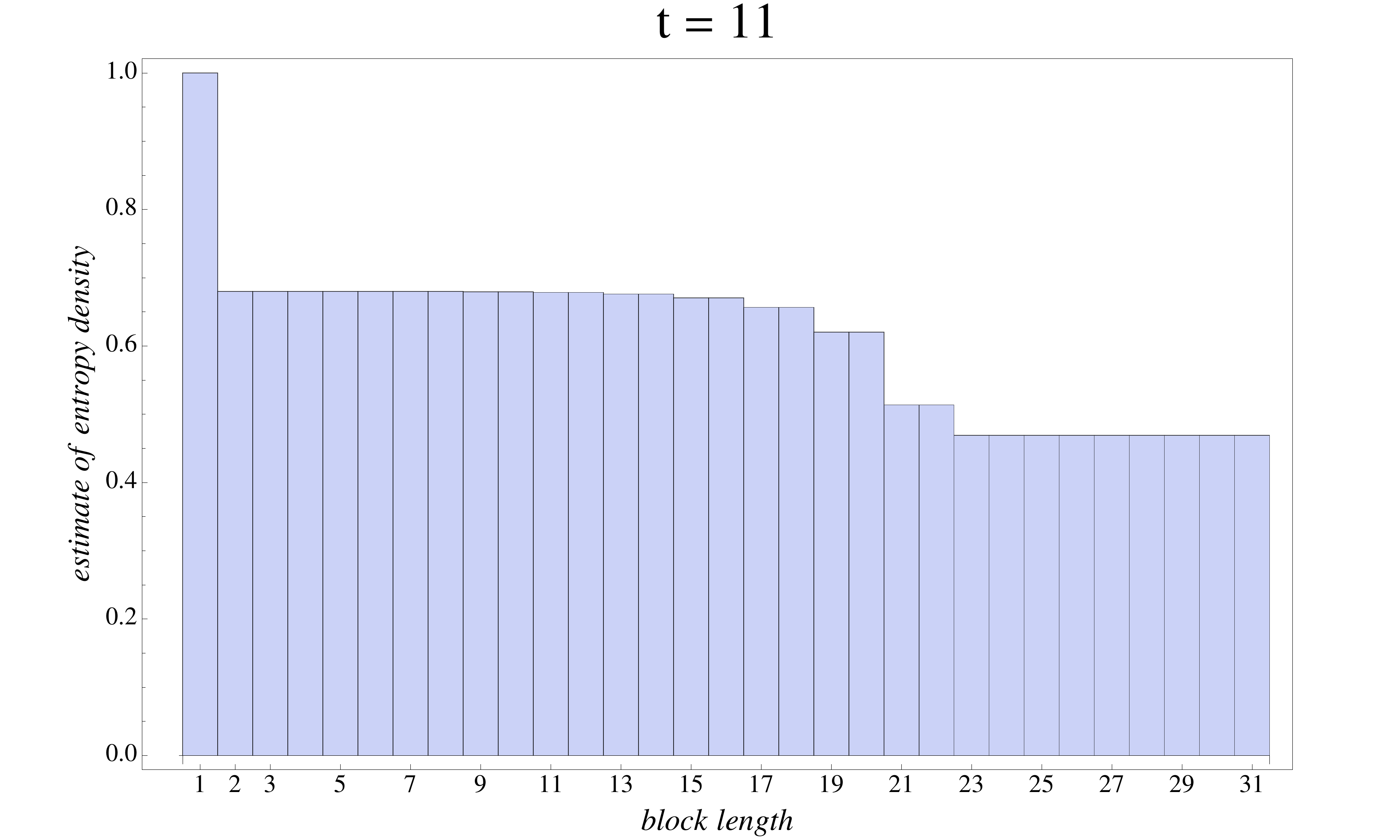}
                \caption{}
                \label{fig:st11}
        \end{subfigure}
        \caption{Temporal evolution of the of the entropy density estimate $h_m$ for block of length $m$ (with $1 \le m \le 31$) at different time steps $t$. For finite time $t$, if sufficiently long blocks $m$ are used for the entropy density estimate $h_m(t)$, then the correct entropy density $h$ is found. But, as time goes on, any finite block length estimate will converge to the equilibrium entropy density, $h_\text{eq}$. (In this calculation of $h_m(t)$ we have not included the constant contribution of $\log 2$ to the pair correlation information $k_2(t)$ that comes from the regular periodic pattern of updating and quiescent states.)}\label{fig:s-estimate}
\end{figure}

%This then means that the rate $2\zeta$ of linear increase in average correlation length, or excess entropy, as expressed by Proposition \ref{Prop:eta-new} is given by the difference in entropy density between the Q2R system and the corresponding equilibrium,
%
%\begin{align}
%\zeta = S[P^{(2)}] - S[P^{(1)}] = h_\text{eq} - h \;.
%\end{align}
%

\subsection{Approaching the equilibrium characteristics}
We have shown that the magnetisation converges to 0 exponentially, see Proposition \ref{magnetisation}, reflecting that $p(\uparrow \uparrow)$ and $p(\downarrow \downarrow)$ converge to the same value. Since $p(\uparrow \downarrow)$ is conserved, this implies that the estimate of the entropy density based on blocks of length 2, $h_2(t)$, approaches the equilibrium entropy density, $h_\text{eq}$,
\begin{align}
h_2(t) \rightarrow h_\text{eq} \; \text{  as  } t \rightarrow \infty \;.
\end{align}
But the convergence towards the equilibrium distribution characteristics goes beyond the pairs of symbols. Even if we would estimate the entropy density by using the $m$-length block statistics, also that would converge towards $h_\text{eq}$. This follows from the following observations.

We can calculate the entropy density estimate $h_m$ from $\log 4 - \sum_{j=1}^m k_j(t)$, where $k_j(t)$ are the correlation information contributions, see Eq.~(\ref{h_m}). This can be rewritten as
\begin{align}
h_m(t) = h_2(t) + \sum_{j=3}^m k_j(t) \;, 
\end{align}
where we have used that $h_2(t) = \log 4 - k_1(t) - k_2(t)$.

From Propositions \ref{Prop:k_2m}, \ref{Prop:k_2m+1}, and \ref{Prop:k_3}, we see that for finite $m>2$ all correlation information terms $k_m(t)$ will eventually decay towards $0$. This means that for finite $m$, the entropy density estimate, $h_m(t)$, converges to the equilibrium entropy density value, $h_\text{eq}$,
\begin{align}
h_m(t) \rightarrow h_\text{eq} \; \text{  as  } t \rightarrow \infty \;. \label{h_m-convergence}
\end{align}
% 
%This means that, if we choose a finite block length $M$, we have, after $\sim M/2$ time steps, an exponential convergence of all block probability distributions $P_m$ (with $m\le M$) towards the corresponding equilibrium ones. 

This process that brings us towards equilibrium characteristics for any finite block length is illustrated in Figs. \ref{fig:s-estimate} and \ref{fig:corr-info}. 

In Fig.~\ref{fig:s-estimate}, the entropy density estimate, $h_m(t)$, as a function of the block length $m$, is depicted for the first four time steps as well as for time steps $t=10$ and $t=11$. Here it is clear that, for the first time steps, we will be able to detect the correct entropy density of the system, but as time proceeds the finite length estimates $h_m(t)$ will converge towards a larger one, i.e., the equilibrium entropy density as we have shown above in Eq.~(\ref{h_m-convergence}).

The corresponding picture for the contributions to the correlation information, as a function of block length $m$, is shown for the same time steps $t$ as above in Fig.~\ref{fig:corr-info}. We note that, even though correlation information in total is conserved, a certain part of it is transferred to longer and longer blocks, as is stated by Proposition~\ref{Prop:k_2m+1}. Only one contribution remains at short length scales, $k_2(t)$, which is determined by the 0 magnetisation and the energy constraint. The figure clearly illustrates that all correlation information, except $k_2(t)$, will eventually become undetectable if finite block statistics is used which explains why finite length estimates of $h_m(t)$ converges to $h_\text{eq}$.

\begin{figure}[h]
        \centering
        \begin{subfigure}[b]{0.48\textwidth}
                \includegraphics[width=\textwidth]{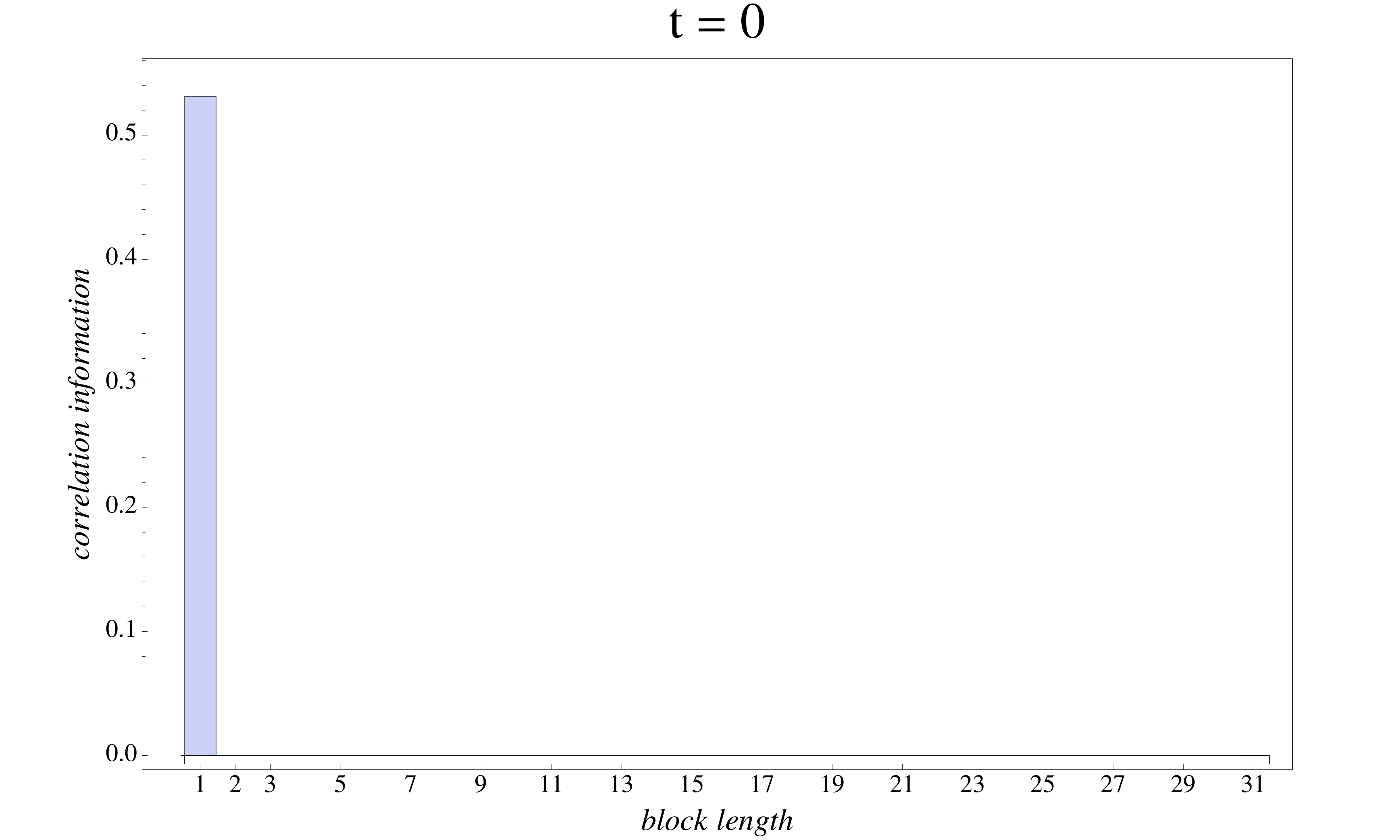}
                \caption{}
                \label{fig:t0}
        \end{subfigure}
        \begin{subfigure}[b]{0.48\textwidth}
                \includegraphics[width=\textwidth]{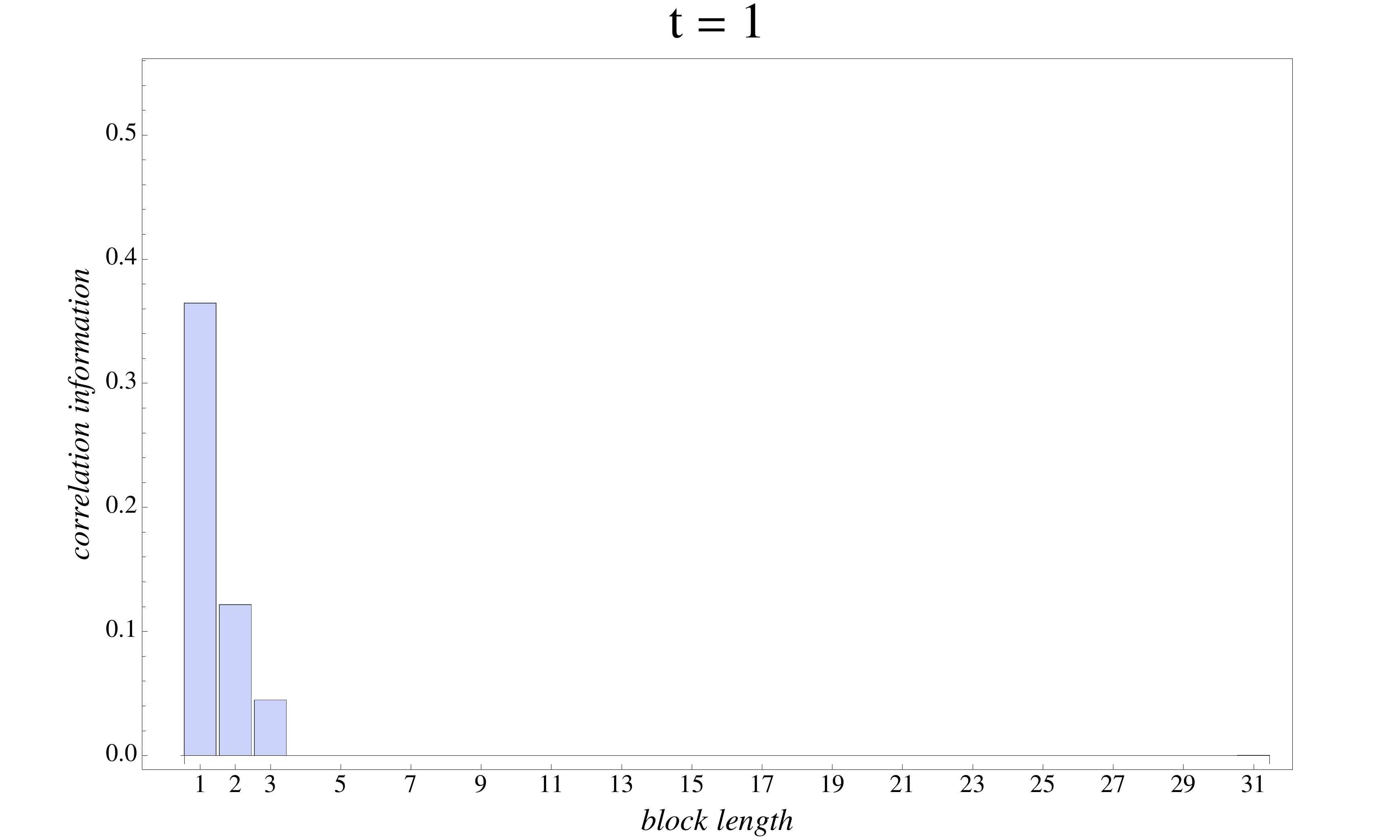}
                \caption{}
                \label{fig:t1}
        \end{subfigure}
        \\[12pt]
        \begin{subfigure}[b]{0.48\textwidth}
                \includegraphics[width=\textwidth]{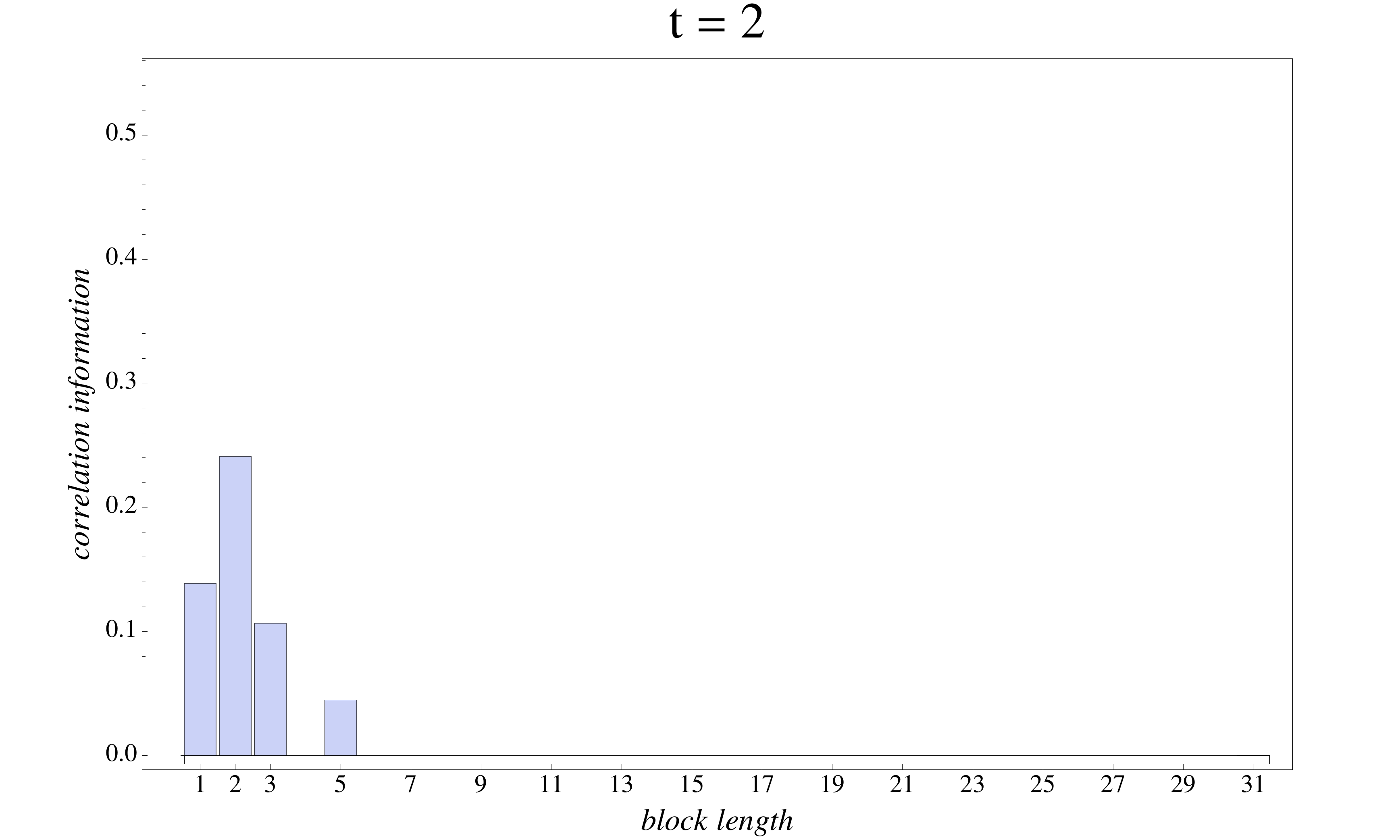}
                \caption{}
                \label{fig:t2}
        \end{subfigure}
        \begin{subfigure}[b]{0.48\textwidth}
                \includegraphics[width=\textwidth]{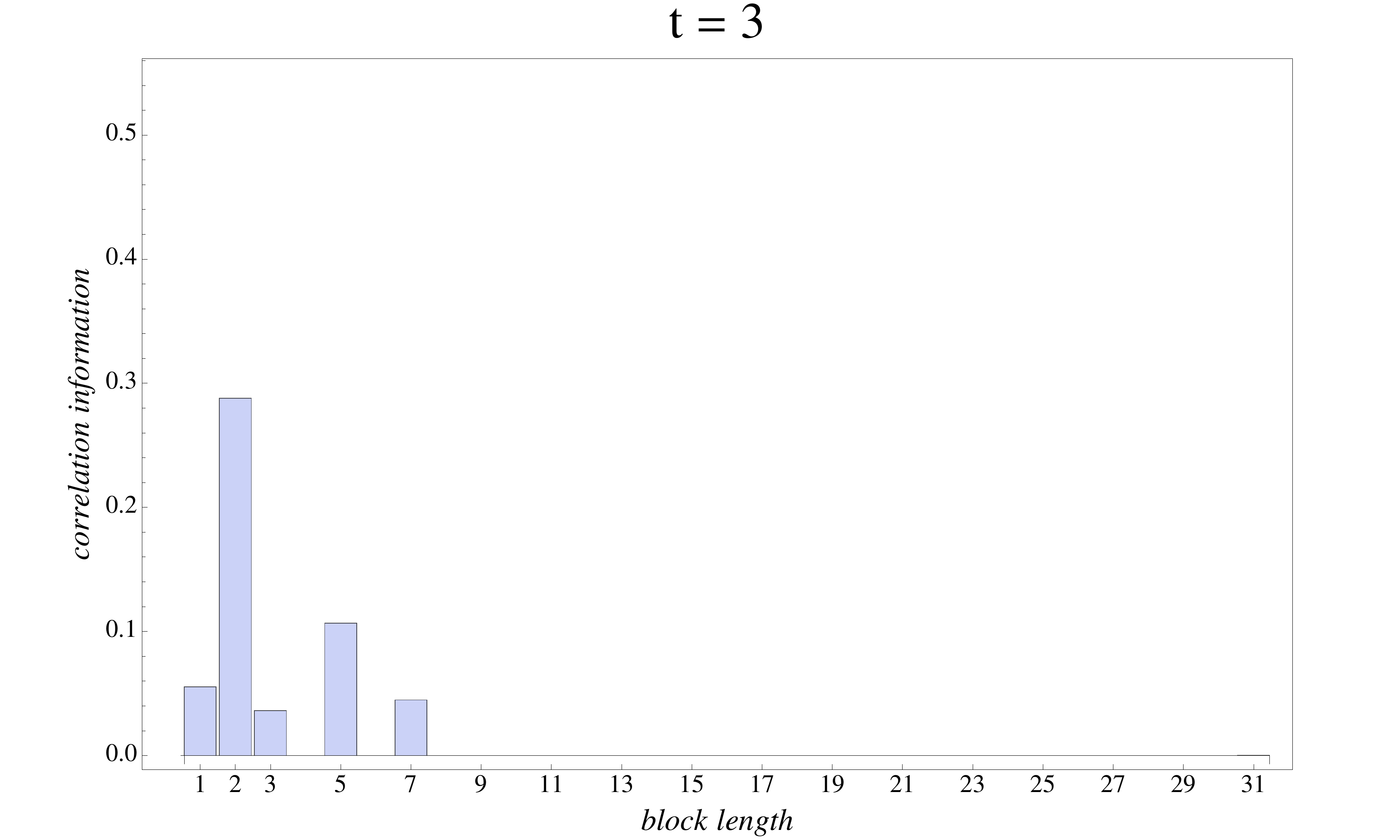}
                \caption{}
                \label{fig:t3}
        \end{subfigure}
        \\[12pt]
        \begin{subfigure}[b]{0.48\textwidth}
                \includegraphics[width=\textwidth]{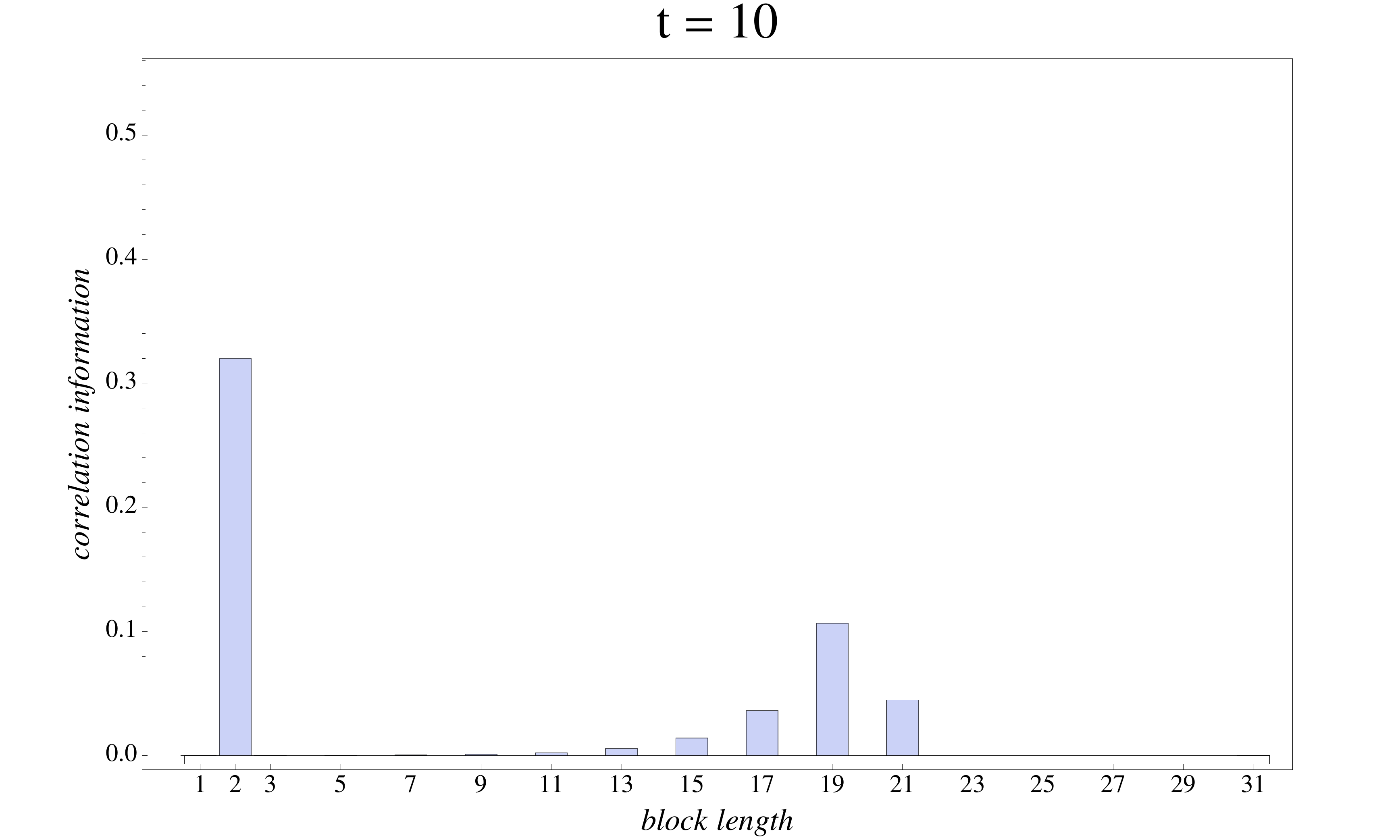}
                \caption{}
                \label{fig:t10}
        \end{subfigure}
        \begin{subfigure}[b]{0.48\textwidth}
                \includegraphics[width=\textwidth]{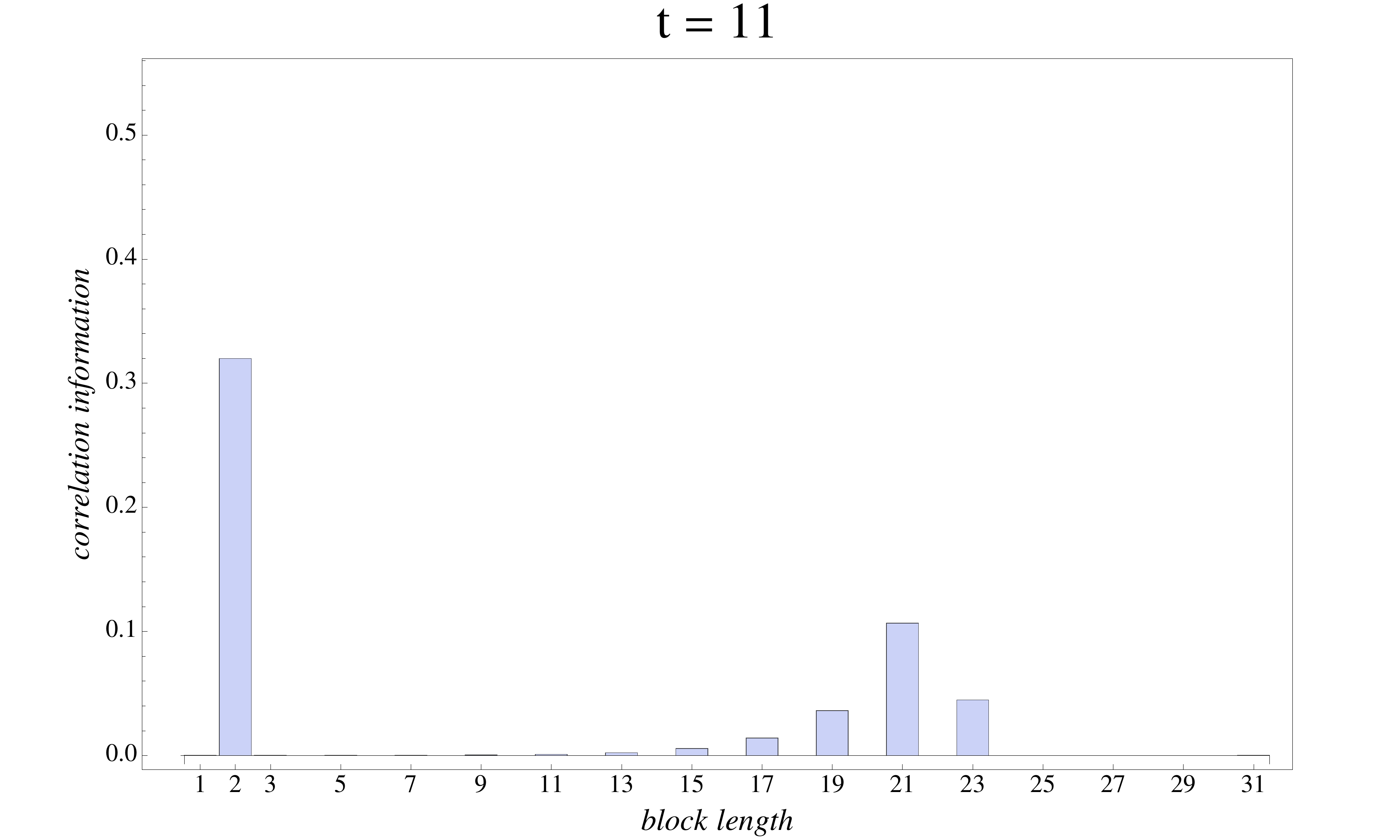}
                \caption{}
                \label{fig:t11}
        \end{subfigure}
        \caption{Temporal evolution of the correlation information $k_m(t)$, for $1 \le m \le 31$, at different time steps $t$. The total correlation information is conserved, but a certain part is transferred to longer and longer blocks. (Here we have subtracted the constant contribution of $\log 2$ to the pair correlation information $k_2(t)$ that comes from the regular periodic pattern of updating and quiescent states.)}\label{fig:corr-info}
\end{figure}

\section{Discussion}
\label{sec:Discussion}
In this paper we have investigated how, and in what sense, a microscopically reversible process can bring a system towards equilibrium. We have considered spatial configurations as generated by stationary ergodic stochastic processes, which has allowed us to study the characteristics of infinite systems directly from the beginning. Moreover, due to the ergodicity, a single microstate can be considered as a typical representative of the whole ensemble. We note that the reversible dynamics, given by the Q2R rule, implies that the (spatial) entropy density is conserved. 
We assume an initial state of independent spins generated by a Bernoulli process and non-zero magnetisation, i.e., $p(\uparrow,t=0) \neq \tfrac{1}{2}$. We have shown that under the reversible Q2R dynamics the system converges exponentially towards zero magnetisation --- the equilibrium value.

The analysis of the conditional entropies $h_m(t)$, the correlation information $k_m(t)$, and the excess entropy $\eta(t)$ as it is depicted in Figs.~\ref{fig:s-estimate} and \ref{fig:corr-info} provides a clear picture of how the approach towards equilibrium characteristics can be understood: The loss of local information, i.e., short-length correlation information, and the corresponding increase of local entropy, i.e., short-length estimates of entropy density $h_m(t)$, is compensated for by building up long-range correlations. The dynamics leads to two different kinds of correlation information. First, there is the ``thermodynamic'' pair correlation information, $k_2(t) \rightarrow k_{2,\text{eq}}$, i.e., the mutual information between neighbouring spins, which characterizes the thermodynamic equilibrium. This term would also arise in a stochastic dynamics of thermalization such as the Glauber dynamics. This term is directly related to the equilibrium value of the thermodynamic entropy density which is given by the plateau of the conditional entropy $h_m(t)$, for large $t$ and $m \ll 2t+1$ and, from Eq.~(\ref{h_m-convergence}), this can be expressed formally as 
\begin{equation}
h_\text{eq} = \lim_{m \to \infty} \lim_{t \to \infty} h_m(t) \;. \label{mt-limit}
\end{equation}
Second, there is a correlation information quantity that directly reflects the reversible nature of the microscopic dynamics: the non-zero terms for $k_m(t)$, with dominating contributions around $m=2t-1$ and $t \gg 1$. The spatial distance on which these dependencies occur increases linearly with time which leads to a linearly increasing excess entropy. These terms are directly related to the difference between the thermodynamic or equilibrium entropy density and the entropy density of the studied system. And this implies that taking the limits of Eq.~(\ref{mt-limit}) in the different order we get the lower entropy density as determined by the initial state,
\begin{equation}
h= \lim_{t \to \infty} \lim_{m \to \infty} h_m(t)  < h_\text{eq} \;.
\end{equation}
The difference between these two expression, i.e., $h_\text{eq} - h$, is the entropy increase when an initial non-equilibrium state is brought to equilibrium. As we have formally shown in this paper, this entropy increase is the amount of correlation information that is being spread out over increasing distances in the time evolution, leading to a linearly increasing information-theoretic correlation length.
Note also that this difference, $h_\text{eq} - h = S[P^{(2)}]-S[P^{(1)}]$ is the linear rate by which the information-theoretic correlation length (or the excess entropy) increases over time as stated in Proposition~\ref{Prop:eta-new}.

\bibliographystyle{spmpsci}
\bibliography{Q2R}

\end{document}